\documentclass[10pt, final, journal, letterpaper, twocolumn]{IEEEtran}
\usepackage{amsfonts}
\usepackage[dvips]{graphicx}
\usepackage{times}
\usepackage{cite}
\usepackage{amsmath}
\usepackage{array}
\usepackage{amssymb}

\usepackage{stfloats}
\usepackage{slashbox}
\usepackage{graphicx}
\usepackage{footnote}
\usepackage{booktabs}
\usepackage{array}
\usepackage{subfig}
\usepackage{algorithmic}
\usepackage{algorithm}
\usepackage[usenames,dvipsnames]{xcolor}
\title{An Auction Approach to  Distributed Power Allocation for Multiuser Cooperative Networks}

\author{Yuan Liu,~\IEEEmembership{Student~Member,~IEEE}, Meixia
Tao,~\IEEEmembership{Senior~Member,~IEEE}, and Jianwei
Huang,~\IEEEmembership{Senior~Member,~IEEE}
\thanks{Manuscript received December 26, 2011; revised April 22, 2012 and August 8, 2012; accepted October 14, 2012. The Editor coordinating the review of this paper and approving it for publication was Prof. Homayoun Yousefi'zadeh.}
\thanks{Y. Liu and M. Tao are with the Department of
Electronic Engineering at Shanghai Jiao Tong University, Shanghai,
200240, P. R. China. Email: \{yuanliu,
mxtao\}@sjtu.edu.cn.}\thanks{J. Huang is with the Department of
Information Engineering at The Chinese University of Hong Kong, Hong
Kong. Email: jwhuang@ie.cujk.edu.hk.}\thanks{This work is supported by the National 973 project under grant 2012CB316100, the Joint Research Fund for Overseas Chinese, Hong Kong and Macao Young Scholars under grant 61028001, and the Innovation Program of Shanghai Municipal Education Commission under grant 11ZZ19. This work of J. Huang is supported by the General Research Fund (Project Number 412710) established under the University Grant Committee of the Hong Kong Special Administrative Region, China. This paper was
presented in part at IEEE Global Telecommunications Conference
(GLOBECOM), Houston, Texas, USA, Dec. 2011 \cite{YuanAuctionGC}.}}

\begin{document}
\maketitle

\vspace{-1cm}
\begin{abstract}

This paper studies a wireless network where multiple users cooperate
with each other to improve the overall network performance. Our goal
is to design an optimal distributed power allocation algorithm  that
enables user cooperation, in particular, to guide each user on the
decision of transmission mode selection and relay selection. Our
algorithm has the nice interpretation of an auction mechanism with multiple
auctioneers and multiple bidders.  Specifically, in our proposed
framework, each user acts as both an auctioneer (seller) and a
bidder (buyer). Each auctioneer determines its trading price and
allocates power to bidders, and each bidder chooses the demand from
each auctioneer. By following the proposed distributed algorithm,
each user determines how much power to reserve for its own
transmission, how much power to purchase from other users, and how
much power to contribute for relaying the signals of others. We
derive the optimal  bidding and pricing strategies that maximize the
weighted sum rates of the users. Extensive simulations are carried out to verify our proposed
approach.

\end{abstract}


\newtheorem{theorem}{Theorem}
\newtheorem{remark}{Remark}
\newtheorem{proposition}{Proposition}
\renewcommand{\algorithmiccomment}[1]{// #1}

\section{Introduction}
\setlength\arraycolsep{2pt}

By exploiting the inherent broadcast nature of wireless radio waves,
users can cooperate and improve the  network throughput and energy
efficiency in wireless networks
\cite{Laneman03,Laneman04,Sendonaris1,Sendonaris2}. Although the
performance of small-scale  cooperative communications has been
extensively studied from an information theoretic perspective, there
still exist many open problems of realizing the full potential of
cooperative communication schemes in practical large-scale networks.

In this paper, we design a distributed resource allocation framework
for cooperative communications that addresses several key practical
challenges.
First of all, forwarding other users' packets consumes valuable
resources (e.g., battery energy, transmission slots or bandwidth)
and may degrade a user's own performance. Therefore, we need to
design a mechanism that guides distributed users to
cooperate.
Second, determining whether and how to perform cooperative
transmission depend on channel conditions between users and can be
complicated.
Third, when cooperative communication is desirable, there can be
more than one user that is suitable to serve as the relay. Thus, we
need to decide how to select the transmission mode (direct or
relay transmission) and the associated relay node(s) for each
source node.
Fourth, if a user decides to help others relaying the messages, it
needs to balance the resource (such as power, bandwidth, and time
slots) reserved for itself and the resource provided for others.
Therefore, we need to design a mechanism so that each user captures
the optimal tradeoff of resource allocation.

Previous work on resource allocation in cooperative networks fall
into two categories: centralized (e.g., \cite{Yu, Himsoon, Tang,
Kim, Le, YuanTW, YuanTCOM}) and distributed (e.g.,
\cite{Bletsas,Savazzi,Sergi,Wang,Huang,Ren}). The centralized schemes
often require global channel state information (CSI), thus are often
not scalable due to the large signaling overhead. Distributed
schemes based on local information and limited message passing are
thus more favorable in practical systems. There have been several
results on distributed resource allocation in cooperative networks.
For example, a distributed power allocation algorithm for
single source-destination pair multi-relay networks was presented in
\cite{Wang} based on the Stackelberg game model, where a source is
modeled as a buyer and relays are modeled as sellers. The power
allocation is only performed by the relays but not by the source.
The authors in \cite{Huang} studied an auction-based distributed
power allocation scheme for the case where multiple
source-destination pairs are assisted by a fixed single relay node.
Therein, the relay acts as the auctioneer and the sources act as the
bidders. Like in \cite{Wang}, \cite{Huang} also adopts relay-power
allocation assuming that the sources' power are fixed.
Authors in \cite{Ren} investigated distributed power allocation in a
multiple source-destination pairs and single relay network, where
the relay sets prices and multiple sources act as a non-cooperative
game. This work assumed that the relay's power is fixed and the
power optimization is done by the sources.

In this paper, we propose a new auction-based power allocation
framework for multi-user cooperative networks, with the objective of
maximizing the weighted sum rates of the users.
Specifically, we design a distributed power allocation algorithm which has the nice interpretation of a  multi-auctioneer multi-bidder power
auction, in which each user acts as both an auctioneer and a bidder.
Each auctioneer independently announces its trading price and sells
power, and each bidder dynamically decides whether to buy, from
which auctioneer(s) to buy, and how much to buy.
By following the proposed distributed mechanism, each user can
achieve the optimal resource utilization and maximize the system
weighted sum rates in a fully distributed fashion.

Our paper distinctly differs from the previous work \cite{Wang,
Huang, Ren} in threefold: 1) unlike those work considering the
\emph{relay-assisted} cooperative communications with dedicated
relays, we study the user cooperation scenario where users cooperate
and help  each other;
2) unlike these results only optimize relay-power \cite{Wang, Huang}
or source-power \cite{Ren}, we optimize the power allocation of each
transmitting node to tradeoff the resource consumption of
transmitting its own traffic and forwarding other nodes' traffic;
3) by employing distributed multiple-input multiple-output (MIMO)
techniques (e.g., \cite{Jing,Yiu}) for relay transmission,
our proposed power allocation scheme implicitly incorporates both
relay selection and transmission mode selection between
direct transmission and relay transmission, which enables each
user to decide whether to cooperate, whom to cooperate with, and how
to cooperate in a distributed fashion through a unified framework.
 We also discuss some implementation issues of the proposed
distributed mechanism in practical wireless networks, including
synchronization, channel estimation, interaction procedure, and
step-size selections.

The remainder of this paper is organized as follows. Section II
describes the system model and the problem formulation. Section III
presents the details of the proposed multi-auctioneer and multi-bidder
power auction. Extensive simulations are provided in Section IV.
Finally we conclude the paper in Section V.

\section{System Model and Problem Formulation}

The network consists of $K$ cooperative users, each of which has its
own source node and destination node. Different users have distinct
source nodes, but may share the same destination node. The users can
cooperate with each other, thus the source node of a user can be a
transmitter and a forwarder simultaneously. The main idea of
cooperative transmission is to utilize the extra power available at
those source nodes, by enabling them to act as relays for those
sources who are far from their destinations. Such cooperative
transmission can be regarded as a distributed multiple-antenna
system.
Each node is subject to the half-duplex constraint, so that it
cannot transmit and receive at the same time over the same channel. It is possible, however, for the node to transmit and
receive simultaneously on different channels.

\begin{figure}[t]
\begin{centering}
\includegraphics[scale=.5]{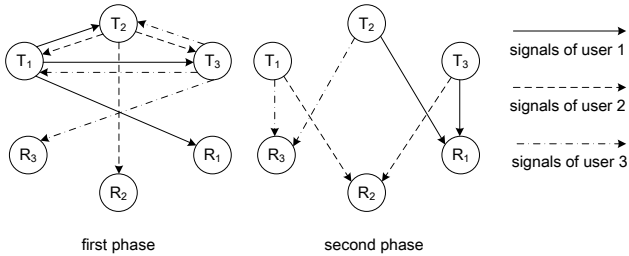}
\vspace{-0.1cm}
 \caption{An example of the cooperative communication
with two users (transmitter-receiver pairs). ${\rm T_i}$ and ${\rm
R_i}$ denote the transmitter and receiver of user $i$,
respectively.}\label{fig:tran}
\end{centering}
\vspace{-0.3cm}
\end{figure}

Let $\mathcal {K}=\{1,2,...,K\}$ be the set of users. Each user is allocated an orthogonal channel to transmit its own data information.\footnote{It is assumed that the spectrum is equally divided based on the number of users.} Without loss of generality, we assume that channel $i$ is allocated to user $i$.  The cooperative transmission
includes two phases, with an example shown in Fig.~\ref{fig:tran}. In
the first phase, the source node of user $i$ ($1\leq i \leq K$)
transmits the message to its own destination on channel $i$ and
listens on all other channels. In the second phase, some source nodes
(who have the proper channel conditions and extra transmission power)
form a distributed MIMO system and simultaneously
help to relay the signals of user $i$ on channel $i$ by using distributed space-time
codes \cite{Jing,Yiu}.

The $K\times K$ matrix $\boldsymbol{p}$ denotes the transmission power, where
$p_{j,i}$ denotes the amount of power that the source node of user
$j$ contributes to forwarding user $i$'s information for $i\neq j$,
or the amount of power user $i$ consumes for its own transmission
for $i= j$. The sum of the $j$th row of $\boldsymbol{p}$ represents
the total power consumption of user $j$, which is subject to a peak
power constraint $\overline{p}_j$.
Note that this power matrix  implicitly accounts for transmission
mode and relay node adaptations. For instance, if $\boldsymbol{p}$
is a diagonal matrix, then only direct transmission is selected;
full user cooperation is selected if all the elements of
$\boldsymbol{p}$ are non-zero.

Let $\boldsymbol{R}$ be a $K\times 1$ vector whose element $R_i$
denotes the achievable rate of user $i$ at a given power allocation
vector $\{p_{j,i}\}_{j=1}^K$. This achievable rate definition can
accommodate different distributed MIMO techniques, such as the nonregenerative based
amplify-and-forward (AF) and regenerative based decode-and-forward (DF) relay strategies. The detailed
expression of rate $R_i$ will be given later.

The global network objective is to allocate the power on each source
node in order to maximize the weighted sum rates of the users. The optimization problem can be formulated as
follows (\textbf{P1}):

\begin{eqnarray}
\textbf{P1:}~~~~~~~{\rm max} && \sum_{i \in \mathcal {K}}w_iR_i \\
s.t. && \sum_{i\in \mathcal {K}}p_{j,i} \leq \overline{p}_j,~\forall j\in\mathcal {K}  \label{eqn:c1}\\
{\rm variables}&& \boldsymbol{p}\succeq \mathbf{0}, \label{eqn:c2}
\end{eqnarray}
where $w_i$ is the weight that represents the priority of user $i$.

If date rate $R_i$ for each user $i$  is a concave function of
the power vector $\{p_{j,i}\}_{j=1}^K$, then the objective function
is concave as any positive linear combination of concave functions
is concave. Moreover, constraint (\ref{eqn:c1}) is convex and
constraint (\ref{eqn:c2}) is affine. Hence the feasible set of this
optimization problem is convex. Therefore, \textbf{P1} is a convex
optimization problem and there exists a globally optimal solution. The
Lagrangian of \textbf{P1} is given by:
\begin{equation}
L\left(\boldsymbol{p},\boldsymbol{\lambda}\right) = \sum_{i\in
\mathcal {K}} w_iR_i - \sum_{j\in\mathcal
{K}}\lambda_j\left(\sum_{i\in\mathcal
{K}}p_{j,i}-\overline{p}_j\right),
\end{equation}
where
$\boldsymbol{\lambda}=(\lambda_1,\lambda_2,...,\lambda_K)\succeq 0$
are the Lagrange multipliers related to the power constraints.
Applying the Karush-Kuhn-Tucker (KKT) conditions \cite{Boyd}, we
obtain the following necessary and sufficient conditions for the
optimal primal variables $\boldsymbol{p}^*$ and dual variables
$\boldsymbol{\lambda}^*$:
\begin{eqnarray}
\frac{\partial L\left(\boldsymbol{p}^*,\boldsymbol{\lambda}^*\right)}{\partial p_{j,i}^*} &\leq& 0, \forall i, j \in \mathcal {K},\label{eqn:a1}\\
\lambda_j^*\left(\sum_{i\in \mathcal {K}}p_{j,i}^*-\overline{p}_j\right) &=& 0, \forall j\in \mathcal {K},\label{eqn:a2}\\
\sum_{i\in \mathcal {K}}p_{j,i}^* &\leq& \overline{p}_j, \forall j\in \mathcal {K},\label{eqn:a3}\\
\boldsymbol{p}^*\succeq \mathbf{0},~\boldsymbol{\lambda}^* &\succeq
& \mathbf{0}. \label{eqn:a4}
\end{eqnarray}
Note that the KKT conditions imply that $w_iR'_i(p_{j,i}^*) = \lambda_j^*$ if $p_{j,i}^*>0$, and $w_iR'_i(p_{j,i}^*) \leq \lambda_j^*$ if $p_{j,i}^*=0$.


Next we study the concavity of the achievable rates $\boldsymbol{R}$
with respect to the power allocation matrix $\boldsymbol{p}$. For an
illustration purpose, we employ the AF relay strategy in this
paper.
We model the wireless fading environment by the large-scale path
loss, shadowing, and small-scale Rayleigh fading. The
additive white Gaussian noises (AWGN) at all users are assumed to be
independent circular symmetric complex Gaussian random variables,
each having zero mean and unit variance.

In the first phase of the cooperative communication, each source
node $i\in \mathcal {K}$ broadcasts its signal $x_i$ to all other
source nodes for all $ j\neq i$ and its destination. Thus the
received signals at its destination and node $j$ from $i$ are given
by, respectively
\begin{equation}
y_{i} = \sqrt{p_{i,i}}h_{i,i}x_{i} + n_i,
\end{equation}
and
\begin{equation}
y_{i,j} = \sqrt{p_{i,i}}h_{i,j}x_{i} + n_j,~\forall j\neq i,
\end{equation}
 where $n_i$ is the AWGN at destination $i$, $h_{i,j}$ is the
channel gain from source $i$ to destination $j$, $\forall i,j$.

In the second phase, all other users participate in the cooperative transmission of user $i$ in the form of AF based distributed space-time coding. Specifically, each user $j \in \mathcal {K}\setminus i$ multiplies its received signal from user $i$ over a certain time duration, denoted as $\boldsymbol{y}_{i,j}$, with a distributed space-time code matrix $\boldsymbol{A}_j$ and then forwards the coded signal to the destination using power $p_{j,i}$.
The received signal vector at user $i$' destination is
\begin{equation}\label{eqn:zi}
\boldsymbol{z}_{i} = \sum_{\forall j\in{\mathcal{K}}\setminus
i}\sqrt{p_{j,i}}h_{j,i}\boldsymbol{A}_j\boldsymbol{y}_{i,j} + n_i.
\end{equation}
Here, the distributed space-time code matrices $\{\boldsymbol{A}_j\}$ are chosen carefully so that the full diversity order can be achieved \cite{Jing}. By
employing the maximal-ratio combining for the direct and cooperative
links, the achievable rate of user $i$ can be written as
{\small
\begin{eqnarray*}\label{eqn:rate}
R_i = \frac{1}{2}{\rm log}_2 \Bigg(1 + p_{i,i}|h_{i,i}|^2
+\sum_{\forall j\in{\mathcal{K}}\setminus i}
\frac{p_{i,i}|h_{i,j}|^2p_{j,i}|h_{j,i}|^2}{1 + p_{i,i}|h_{i,j}|^2 +
p_{j,i}|h_{j,i}|^2}\Bigg).
\end{eqnarray*}
}

Note that $R_i$ is not jointly concave in $p_{i,i}$ and $p_{j,i}$. To make the analysis tractable, we
adopt the following upper bound approximation:
\begin{equation}\label{eqn:approx}
R_i \approx \frac{1}{2}{\rm log}_2 \left(1 + p_{i,i}|h_{i,i}|^2
+\sum_{\forall j\in{\mathcal{K}}\setminus i}
\frac{p_{i,i}|h_{i,j}|^2p_{j,i}|h_{j,i}|^2}{p_{i,i}|h_{i,j}|^2 +
p_{j,i}|h_{j,i}|^2}\right),
\end{equation}
assuming that the signal amplified and forwarded by the relays is in
the high SNR regime. Such upper bound is tight and commonly used for
AF rate expression in the literature (e.g., \cite{Tang}) for the
single-relay case.
It can also be proved that  $R_i$ in \eqref{eqn:approx} is strictly concave by evaluating
the Hessian matrix \cite{Boyd}.
Note that in \eqref{eqn:zi}-\eqref{eqn:approx} we use
the general notation, $\forall j\in{\mathcal K}\setminus i$,
for the relay index $j$ with respect to user $i$. This does not change
the rate results nor the concavity of $R_i$, because for a user
$j$ who is not involved in the cooperation we can simply let
$p_{j,i}=0$. We can also view this as the relay selection decision of user $i$.
Moreover, if $p_{j,i}=0$ for all $j\in{\mathcal K}\setminus i$, then
the transmission mode of user $i$ becomes direct transmission,
otherwise it uses relay transmission mode. In other words,
our proposed power allocation framework implicitly involves
transmission mode selection between direct transmission and relay
transmission.

In the next section, we propose an auction based algorithm to
achieve the optimal solution of \textbf{P1} in a distributed
fashion.

\section{Auction-Based Distributed Power Allocation}

Auction theory \cite{Krishna} has been viewed as an efficient method
to allocate wireless resource in different scenarios, e.g., rate
control \cite{Kelly}, spectrum allocation \cite{Sun,Yang,Chen},
spectrum access \cite{Zhou,Stanojev,Niyato1}, and spectrum sharing
\cite{Niyato,Huang2}.
Most of these schemes are based on a centralized auctioneer (or a
seller) since there is a single divisible resource to be allocated
among bidders (buyers). Therefore, these auction-based resource
allocation schemes cannot be directly applied to our considered
multi-user cooperation scenario.

In this section, we design a distributed algorithm to solve the problem \textbf{P1}. The distributed algorithm has a nice interpretation of a multi-auctioneer and multi-bidder auction.
We further analyze the convergence and discuss the implementation
issues in practical networks.

\subsection{Multi-Auctioneer Multi-Bidder Mechanism}

\begin{figure}[t]
\begin{centering}
\includegraphics[scale=.7]{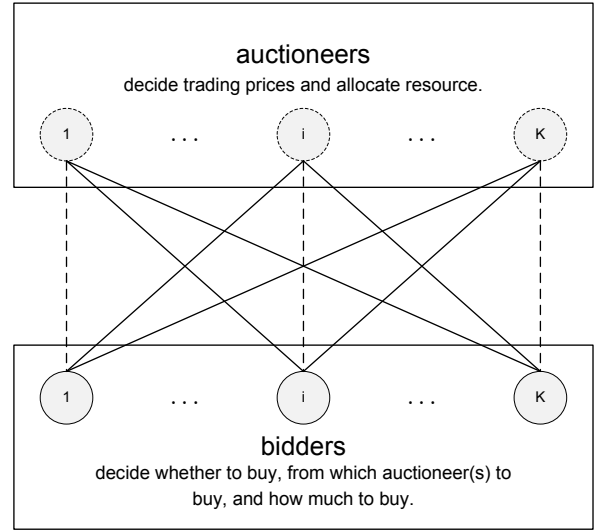}
\vspace{-0.1cm}
 \caption{Illustration of interaction among nodes.}\label{fig:inter}
\end{centering}
\vspace{-0.3cm}
\end{figure}

We achieve efficient  power allocation through a multi-auctioneer
multi-bidder power auction. Each user has two roles in the auction:
an auctioneer and a bidder. The interaction of the
users is illustrated in Fig. \ref{fig:inter}, in which  each bidder
dynamically decides whether to buy, from which auctioneer(s) to buy,
and how much to buy. In the sequence, each auctioneer independently
announces its trading price and allocates power to the bidders.

In the proposed approach, each user $i$ submits a bid $b_{j,i}$ to
each user $j$. When $i=j$, we assume that user $i$ submits a
\textit{virtual} bid $b_{i,i}$ to itself, since it also acts as an
auctioneer (denoted as dotted lines in Fig.~\ref{fig:inter}).
%
By bidding for its own resource, each user $i$ can determine how
much power $p_{i,i}$ to consume for itself, besides how
much power $p_{j,i}$ (for all $ j\neq i$) it buys from others.

Let $\boldsymbol{\pi}=(\pi_1,\pi_2,\ldots,\pi_K)$ be the price
vector of the auctioneers, and $\boldsymbol{b}$ the bidding matrix
whose element $b_{j,i}$ means the willingness bidder $i$ buys from
auctioneer $j$. By definition,
$\boldsymbol{b}_i=\{b_{j,i}\}_{j=1}^K$, the column $i$ of
$\boldsymbol{b}$, is user $i$'s bidding vector. The multi-auctioneer
multi-bidder power auction consists of two steps:
\begin{enumerate}
\item For a given price vector $\boldsymbol{\pi}$, each user $i\in\mathcal{K}$,
determines its demand vector $\{p_{j,i}\}_{j=1}^K$, then submits the
corresponding bid vector $\{b_{j,i}\}_{j=1}^K$ to auctioneers;

\item For the collected bids $\boldsymbol b$, each auctioneer $j\in\mathcal{K}$,
determines its supply vector $\{p_{j,i}\}_{i=1}^K$ and announces its
price $\pi_j$.
\end{enumerate}
The key challenge is how to design the price vector
$\boldsymbol{\pi}$ and bidding matrix $\boldsymbol{b}$ so that the
outcome of the proposed power auction is equivalent to the optimal
solution of \textbf{P1}.

We introduce the two-side auction rule as follows:
\begin{enumerate}
\item At the side of the bidders, each bidder $i\in\mathcal{K}$,
submits it bid proportionally to the price of auctioneer $j$ and the
power it will purchase from auctioneer $j$, i.e.,
$b_{j,i}=\pi_jp_{j,i}$, $\forall j$. Intuitively, if $p_{j,i}=0$,
bidder $i$ does not bid for auctioneer $j$.
\item At the side of the auctioneers, we adopt Kelly-mechanism \cite{Kelly} such that each auctioneer $j\in\mathcal{K}$, attempts to
maximize the surrogate function $\sum_{i \in \mathcal
{K}}b_{j,i}{\rm log}p_{j,i}$ by allocating power $p_{j,i}$ according the bids $b_{j,i}$. Note that the surrogate function can
be selected arbitrarily as long as it is differentiable, strictly increasing,
and concave in $p_{j,i}$.
\end{enumerate}

In what follows, we describe the two-side auction in detail.

\subsubsection{Bidder-Side Auction}
Each bidder $i$ maximizes its surplus, which is the difference
between the benefit from buying power from auctioneers and its
payments. For the given prices $\boldsymbol{\pi}$, bidder $i$
 first determines its optimal demand according to the
following surplus maximization  (\textbf{Bidder Sub-Problem}):
\begin{equation}\label{eqn:indpro}
\max_{\{p_{j,i}\}_{j=1}^K} S_{i} = w_iR_i - \sum_{j\in \mathcal
{K}}\pi_jp_{j,i}.
\end{equation}

It is not difficult to prove that the surplus function $S_i$ is
jointly concave in $\{p_{j,i}\}_{j=1}^K$, where $R_i$ is a function
of $\{p_{j,i}\}_{j=1}^K$ (as defined in (\ref{eqn:approx})).  Due to
the concavity of $S_i$, bidder $i$ can optimally choose the unique
power vector $\{p_{j,i}^*\}_{j=1}^K$ to maximize its profit.
Then bidder $i$ submits its optimal bids to auctioneers according
its optimal demand and the given prices $\boldsymbol{\pi}$:
\begin{equation}
b_{j,i}^* = p_{j,i}^*\pi_j,~\forall j.
\end{equation}

Differentiating $S_i$ with respect to $p_{j,i}$, we can obtain the
sufficient and necessary first order condition:
\begin{equation}\label{eqn:s}
\frac{\partial S_i}{\partial p_{j,i}^*} = w_iR_i'(p_{j,i}^*) - \pi_j =
0, ~~~\forall i, j.
\end{equation}

Observing the KKT conditions of \textbf{P1}, we
notice that if auctioneers announce their prices as
\begin{equation}
\pi_j^* = \lambda_j^* = w_iR_i'(p_{j,i}^*),~~~\forall i, j,
\end{equation}
the optimal power $\boldsymbol{p}^*$ in the bidder sub-problem is
consistent with the one in \textbf{P1}.

From above we can see that the individual optimum in the Bidder
Sub-problem is also the global optimum if the prices are
appropriately selected.

\subsubsection{Auctioneer-Side Auction}
After introducing the Bidder Sub-Problem, we now turn to the
auctioneers.
Solving the optimal power supply of each auctioneer $j$ can be
formulated as (\textbf{Auctioneer Sub-Problem}):
\begin{eqnarray}\label{eqn:aucsubp}
{\rm max} && \sum_{i \in \mathcal {K}}b_{j,i}{\rm log}p_{j,i} \\
s.t. && \sum_{i\in \mathcal {K}}p_{j,i} \leq \overline{p}_j  \label{eqn:cc1}\\
{\rm variables} && \boldsymbol{p}\succeq \boldsymbol{0}
\label{eqn:cc2}
\end{eqnarray}
The associated Lagrangian can be written as
\begin{equation}
L'_j=\sum_{i \in \mathcal {K}}b_{j,i}{\rm
log}p_{j,i}-\mu_j\left(\sum_{i\in \mathcal {K}}p_{j,i} -
\overline{p}_j\right),
\end{equation}
where $\mu_j$ is the Lagrange multiplier of auctioneer $j$. The KKT
conditions for the Auctioneer Sub-Problem are given by
\begin{eqnarray}
p_{j,i}^* &=& \frac{b_{j,i}}{\mu_j^*}, \forall i \in \mathcal {K},\label{eqn:aa1}\\
\mu_j^*\left(\sum_{i\in \mathcal {K}}p_{j,i}^*-\overline{p}_j\right) &=& 0, \label{eqn:aa2}\\
\sum_{i\in \mathcal {K}}p_{j,i}^* &\leq& \overline{p}_j, \label{eqn:aa3}\\
\boldsymbol{p}^*\succeq \mathbf{0},~\mu_j^* &\geq & 0.
\label{eqn:aa4}
\end{eqnarray}
By comparing the Auctioneer Sub-Problem with \textbf{P1}, one can
see that if $\boldsymbol{\mu}=\boldsymbol{\lambda}$ and bidders
select their bids as follows:
\begin{equation}
b_{j,i}^* = p_{j,i}^* w_iR'_i(p_{j,i}^*),
\end{equation}
then (\ref{eqn:aa1})-(\ref{eqn:aa4}) are equivalent to
(\ref{eqn:a1})-(\ref{eqn:a4}) and the solutions of the Auctioneer
Sub-Problem for all auctioneers coincide with \textbf{P1}.

\subsection{Distributed Algorithm}

We now design a mechanism to realize the multi-auctioneer
multi-bidder power auction in a distributed fashion, in which we
incorporate \textit{primal-dual algorithms} which have been studied
extensively in the literature. The mechanism is executed
iteratively. Formally, we present the detailed mechanism in
Algorithm 1. Each iteration consists of bid update (Algorithm 2),
power allocation, and price update. Note that Algorithm 2
incorporates auctioneer (or relay) selection and transmission mode
selection.
\begin{algorithm}[!t]
\caption{Multi-Auctioneer Multi-Bidder Power Auction}
\textbf{Initialization}. Set the iteration index $t=0$. Set the
direct transmission mode to be the initial state for all nodes,
i.e., $\boldsymbol{p}^{(0)}={\rm
diag}(\overline{p}_1,...,\overline{p}_i,...,\overline{p}_K)$.
Randomly generate a $K\times K$ bid matrix $\boldsymbol{
b}^{(0)}\succeq 0$ and a price vector $\boldsymbol{
\lambda}^{(0)}=(\lambda_1^{(0)},...,\lambda_i^{(0)},...,\lambda_K^{(0)})\succeq
0$.

\textbf{repeat}
\begin{itemize}

\item  $t\leftarrow t+1$.

\item \textbf{Bidder Sub-Problem} in
\eqref{eqn:indpro}:
\begin{itemize}
\item \textit{Bid update}.~~~//~\texttt{Algorithm $2$}
\end{itemize}

\item \textbf{Auctioneer Sub-Problem} in
\eqref{eqn:aucsubp}:
\begin{itemize}
\item \textit{Power allocation}.  Each user $j\in \mathcal {K}$  (as an auctioneer)
independently  allocates the power $p_{ji}^{(t)}$ to user $i$:
\begin{equation}
p_{j,i}^{(t)} = \frac{b_{j,i}^{(t-1)}}{\lambda_{j}^{(t-1)}},~~~{\rm
for} ~i=1,2,...,K. \label{eqn:allo}
\end{equation}

\item \textit{Price update}. Each user  $j\in \mathcal
{K}$ (as an auctionneer) updates its price as:
\begin{equation}
\lambda_j^{(t)} = \lambda_j^{(t-1)} +\epsilon_j \left(\sum_{i\in
\mathcal {K}}p_{j,i}^{(t)} - \overline{p}_j\right),\label{eqn:pd}
\end{equation}
where $\epsilon_j$ is a small constant step-size.
\end{itemize}
\end{itemize}
\textbf{until} The price vector  $\boldsymbol{\lambda}$ converges.
\end{algorithm}

\begin{algorithm}[!t]
\caption{Bid update}
\begin{algorithmic}[1]
\STATE  $t\geq 1$; \FOR{each bidder $i=1:K$} \FOR{each auctioneer $j
= 1:K$}

\IF{$\partial S_i^{(t)}/\partial p_{j,i}^{(t)}>0$ \textbf{or
equivalently} $w_iR'_i\left(p_{j,i}^{(t)}\right)>\lambda_j^{(t-1)}$}
 \STATE Submit bid $b_{j,i}^{(t)} = p_{j,i}^{(t)}
w_iR'_i\left(p_{j,i}^{(t)}\right)$ to auctioneer $j$;
\ELSE \STATE $b_{j,i}^{(t)} =0$ and do not submit bid
to auctioneer $j$;
 \ENDIF \ENDFOR \ENDFOR
\end{algorithmic}
\end{algorithm}

There are several points should be noted. Firstly, from
(\ref{eqn:pd}), one can see that we must have $\sum_{i\in \mathcal
{K}}p_{j,i}^{(t)} = \overline{p}_j$ when $t\to\infty$ so that the
KKT condition (\ref{eqn:aa2}) is satisfied as the price $\lambda_j$
cannot be zero.

Secondly, in the procedure of bid update (Algorithm 2),
it is not necessary for each bidder to submit positive bids to all
auctioneers in each iteration. Specifically, the auctioneer selection depends on
two factors. (\ref{eqn:allo}) shows that the purchased power
is proportional to the submitted bids and inversely proportional to
auctioneers' prices. This means that the higher willingness bidder
$i$ has for auctioneer $j$ and the lower price auctioneer $j$
announces, the more power bidder $i$ can get from auctioneer $j$.
Therefore,  a feasible way of auctioneer selection for bidder $i$ is
to observe how $S_i$ varies with $p_{j,i}$, i.e., observe the sign
of $\frac{\partial S_i}{\partial p_{j,i}}$ or, equivalently,
$w_iR'_i(p_{j,i})-\lambda_j$. Note that $\boldsymbol{p}$ is obtained by
Algorithm 1 in the current iteration, and the price vector
$\boldsymbol{\lambda}$ is obtained in the latest previous iteration.
Let us first discuss the case of $i\neq j$. If $\frac{\partial
S_i}{\partial p_{j,i}}>0$ or $w_iR'_i(p_{j,i})>\lambda_j$, this means
that bidder $i$ can obtain a larger profit by increasing the
purchased power $p_{j,i}$ then bidder $i$ bids for auctioneer $j$,
otherwise auctioneer $j$ should not be selected for bidder $i$. This
is due to that the channel gain $h_{i,j}$ is week such that bidder
$i$ cannot benefit from auctioneer $j$, or many other bidders bid
for auctioneer $j$ so that auctioneer $j$ raises its price, then the
profit bidder $i$ benefits from auctioneer $j$ cannot compensate the
payment bidder $i$ pays to auctioneer $j$. The bidding criterion is
also applicable to the case of $i=j$. If $\frac{\partial
S_i}{\partial p_{i,i}}>0$, this means user $i$ can increase its
profit by increasing the consuming power $p_{i,i}$, due to the
channel gain $h_{i,i}$ or $h_{i,j}$ (for some $j$) is strong then it
has more willingness to consume the power for itself, otherwise user
$i$ prefers to sell its power to other users because it can obtain
higher profit by charging others, rather than consumes the power
itself.

Thirdly, in the proposed auction algorithm, it is assumed that all users are price takers, which means that they do not choose their bids strategically to impact the auctioneers' prices. The assumption is reasonable when there are many bidders such that each bidder's impact on the prices is small, and has been widely used in the literature (e.g., \cite{Wang, Huang, Ren,Kelly,Sun,Yang,Chen,Zhou,Stanojev,Niyato1,Niyato,Huang2}). However, when the number of users is small, then users' price anticipating behavior may change resource allocation, and the system should be modeled as a game between the users and the network. The corresponding solution concept is Nash equilibrium, and it is well known that Nash equilibrium often leads to network performance degradation comparing with a globally optimal solution. Such performance gap is often called the ``price of anarchy" \cite{Johari}. Notice that our goal is to design a distributed algorithm to solve a system-level optimization problem instead of game theoretical analysis.

\begin{proposition}
For any initial condition $(\boldsymbol {p}^{(0)},\boldsymbol
{b}^{(0)}, \boldsymbol {\lambda}^{(0)})$, the proposed
multi-auctioneer multi-bidder power auction globally converges to the
globally optimal point $(\boldsymbol {p}^*,\boldsymbol {b}^*,
\boldsymbol {\lambda}^*)$ as $t\rightarrow \infty$.
\end{proposition}
\begin{proof}
Please see Appendix~A.
\end{proof}

\subsection{Discussion on Implementation Issues}

In this subsection, we show how the proposed multi-auctioneer
multi-bidder power auction can be applied in practical distributed
networks.

First of all, we assume that the users interact each other
synchronously. The synchronization can be implemented at the
head of each transmission packet, in which pilot symbols carry out
the task. This means at the beginning of the auction, every
user (bidder role) bids for power simultaneously, then every user
(auctioneer role) allocates power and announces price
simultaneously.  It is worth noting that our algorithm is also applicable for the case where the users interact each other asynchronously (e.g., there exist some malfunctional users that may not update their actions as frequently as the others) but leads to longer convergence time, which will be detailed in the numerical results. We assume the synchronization here only for reducing the convergence time.

Second, the source of user $i$ needs to know the local information, including $h_{i,i}$,
$h_{j,i}$, and $p_{j,i}$ for all $j\neq i$. Specifically, the destination of user $i$ needs to send the CSI $h_{i,i}$ to its own source, and the source of user $j$ needs to send the CSI $h_{j,i}$ to the source of user $i$ for all $i\neq j$.
Here it is assumed that channel estimation can be implemented at
both source and destination of each user by pilot symbols in the head of
frame, then the CSI can be sent via a feedback channel. Moreover, the
source of user $i$ can know the number of
available relayed nodes and the amount of bidding power
$p_{j,i}$ by itself selecting other users for bidding.

Third, a trading place (i.e., public channel) is needed for
interacting auction information including the initial information
$(\boldsymbol {p}^{(0)},\boldsymbol {b}^{(0)}, \boldsymbol
{\lambda}^{(0)})$, and iterative information  $(\boldsymbol
{p}^{(t)},\boldsymbol {b}^{(t)}, \boldsymbol {\lambda}^{(t)})$. In
each iteration, each user (bidder role) first submits bids and then
determines price and allocates power (auctioneer role) according to
the collected bids. At the end of each iteration, each user makes a
decision and updates its bid and price. Note that an auctioneer's
decision is impacted by bidders' immediate willingness how much to
buy from it, rather than other auctioneers' prices. Thus an auctioneer unnecessarily decodes the prices of others,
though they are announced in an open manner.

Fourth, an appropriate step size for each user is needed. In this paper, we adopt the constant step size rule in \eqref{eqn:pd}.\footnote{Note that there exist other three step size rules, i.e., constant step length, square-summable but not summable, and diminishing step size rules \cite{gradient}. They all can guarantee toward the global optimum but lead to different convergence time, and they have their own advantages and disadvantages \cite{gradient}. How to choose and switch the step size rules in the updating procedure is important in practice, but a detailed discussion seems to be out of the scope of this paper. } Generally
speaking, a larger step size requires less time per iteration but leads to more iterations, while a smaller step size requires more time per iteration but leads to less iterations. In the proposed protocol, if the auction phase ends then
the process turns into the transmission phase. While if any user
does not converge at the end of the auction phase, it will be forced
turn to the transmission phase. In this case, it can randomly make a
decision, which results in the outcome is not optimal and the
performance may degrade substantially. Therefore it is essential
that each user achieves stability in the auction phase. A feasible
way is to enlarge the period of the auction phase, which will reduce
the spectrum efficiency. Here we assume that a slot can be designed
long enough in the system so that the overhead of the auction phase
is negligible.

\section{Simulation Results}

In this section, we evaluate the performance of the proposed
distributed algorithm using simulation. In what follows, we first
consider an example of four users for explicitly illustrating our
proposed multi-auctioneer multi-bidder power auction. Then we
demonstrate the throughput efficiency of the proposed algorithm over
the direct transmission scheme in a network with different number of
users. Without loss of generality, we let $w_i=1$ for all users.

\subsection{A Toy Example of Four Users}

To easily illustrate the details of  the proposed distributed algorithm for multiuser
cooperation, we first consider a toy example with four users. The
system setup for the four-user network is a two-dimensional plane of
size $1\times 1$ km as shown in Fig.~\ref{fig:initial}, where the
destination is located at $(0,0)$, and four nodes are located at
$(0.2,0.5)$, $(0.4,0.3)$, $(0.5,0.8)$, and $(0.8,0.6)$. All nodes
have the same maximum power constraints, i.e., $\overline{p}_i=10$
dB, $\forall i\in \mathcal {K}$. The central frequency is around 5
GHz. The statistical path loss model and shadowing are referred to
\cite{Erceg}, where we set the path loss exponent to be $3.5$ and
the standard deviation of log-normal shadowing is $5.8$ dB. The
small-scale fading between any two nodes is characterized by the
normalized Rayleigh fading.

\begin{figure}[t]
\begin{centering}
\includegraphics[scale=.6]{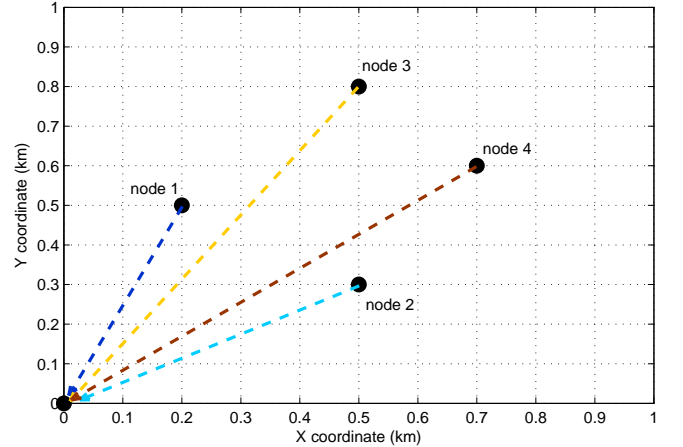}
\vspace{-0.1cm}
 \caption{Initial topology of a network with $4$ user and a destination.}\label{fig:initial}
\end{centering}
\vspace{-0.3cm}
\end{figure}
\begin{figure}[t]
\begin{centering}
\includegraphics[scale=.6]{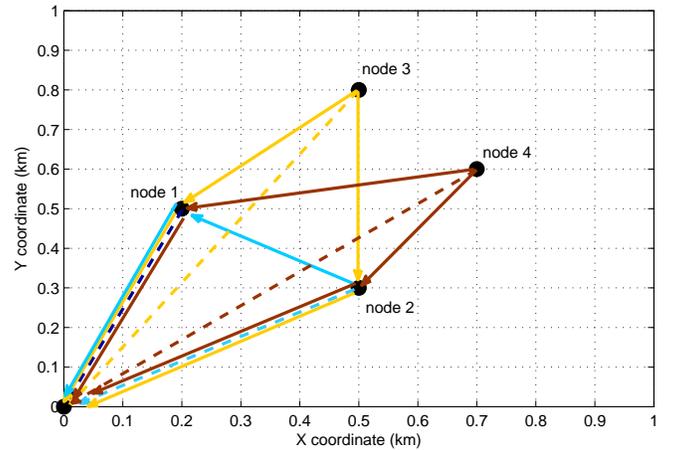}
\vspace{-0.1cm}
 \caption{Final topology of a network with $4$ user and a destination.}\label{fig:final}
\end{centering}
\vspace{-0.3cm}
\end{figure}

Fig.~\ref{fig:initial} presents a topology of the network in the
initial state, where all users transmit information with their
maximum power without cooperation. By implementing the proposed
power auction, the nodes are stimulated to cooperate with each
other. For a given channel realization, the final state of the
network topology is shown in Fig.~\ref{fig:final}, where dash lines
represent direct transmission and solid lines represent relay
transmission, and each color represents the links of one user. One
can see that the nearer a node is to the destination, the more
likely it acts as a relay. For example, as node $1$ is the nearest
to the destination, it forwards the information of other three
nodes, and node $2$ forwards the information of node $3$ and node
$4$. It is also found that node $2$ not only helps node $3$ and node
$4$ but also needs help from node $1$. We further show the power
distribution for each node in Fig.~\ref{fig:power1} and Fig.~\ref{fig:power2}.

\begin{figure}[t]
\begin{centering}
\includegraphics[scale=.55]{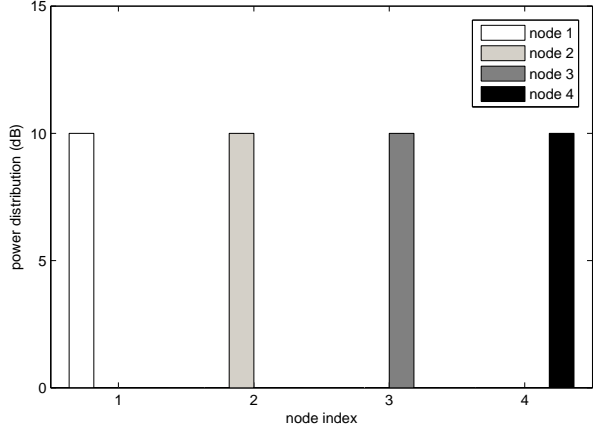}
\vspace{-0.1cm}
 \caption{Initial power distribution among nodes.}\label{fig:power1}
\end{centering}
\vspace{-0.3cm}
\end{figure}
\begin{figure}[t]
\begin{centering}
\includegraphics[scale=.55]{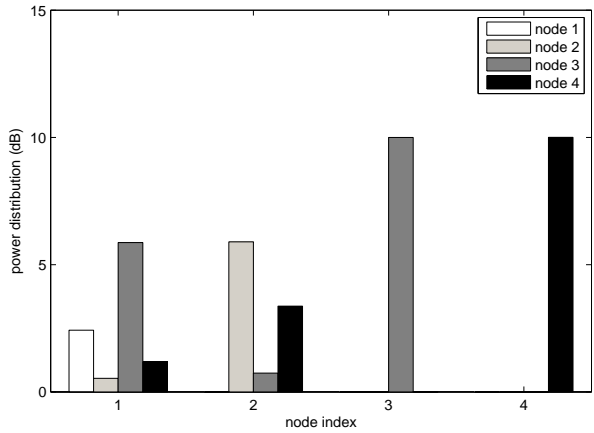}
\vspace{-0.1cm}
 \caption{Final power distribution among nodes.}\label{fig:power2}
\end{centering}
\vspace{-0.3cm}
\end{figure}

\begin{figure}[t]
\begin{centering}
\includegraphics[scale=.55]{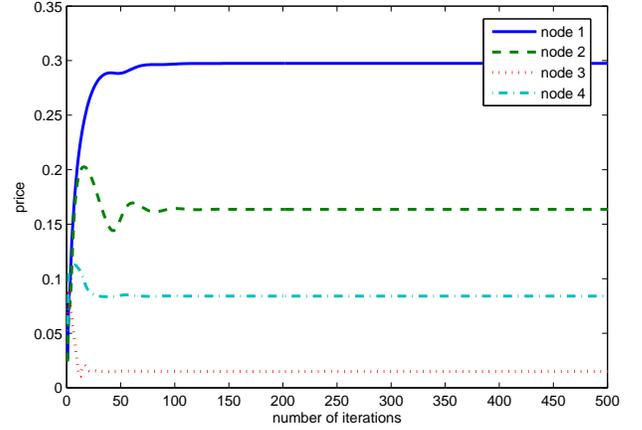}
\vspace{-0.1cm}
 \caption{Convergence speed of prices using $\epsilon_i=10^{-3}$, $\forall i$.}\label{fig:price}
\end{centering}
\vspace{-0.3cm}
\end{figure}
\begin{figure}[t]
\begin{centering}
\includegraphics[scale=.55]{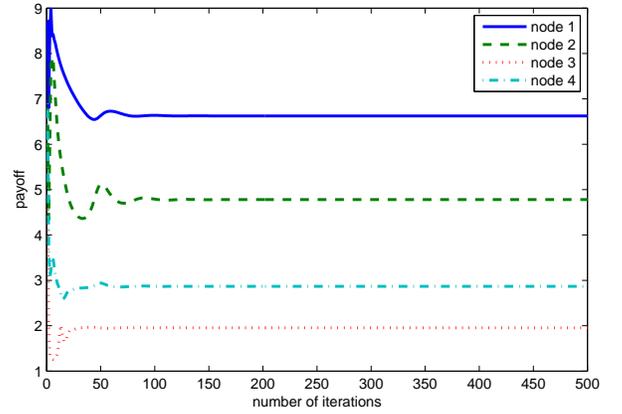}
\vspace{-0.1cm}
 \caption{Convergence speed of payoffs using $\epsilon_i=10^{-3}$, $\forall i$.}\label{fig:payoff}
\end{centering}
\vspace{-0.3cm}
\end{figure}

Fig.~\ref{fig:price} and Fig.~\ref{fig:payoff} show the dynamic
strategies and payoffs of the users, respectively, using step size
$\epsilon_i=10^{-3}$, $\forall i\in \mathcal {K}$. The total payoff
of user $i$ illustrated in the figure is defined as $R_i -
\sum_{j\in \mathcal {K}} p_{j,i}\lambda_j +
\lambda_i\sum_{j\in\mathcal {K}}p_{i,j}$, $\forall i\in\mathcal
{K}$, which is the sum profits of both the auctioneer and bidder.
First of all, we can find that the prices and surplus are convergent
(about $60$ iterations in this example). Secondly, one can see that
the more likely a node acts as a relay, the higher price it
addresses, and the higher surplus it can achieve. Thirdly, for node
$1$ and node $2$ acting as relays, they need more iterations than
node $3$ and node $4$. This can be interpreted as that node $1$ and
node $2$ are in a dilemma whether to transmit themselves or help
others. If they decide to help others, they have to face the
tradeoff between how much power they retain for themselves and how
much power they devote to others. Thus they need more interaction to
make decisions. The observations coincide with the common sense of
economics.

Though we have theoretically proved that the proposed algorithm can
converge to the globally optimal solution when the length of the
auction phase goes to infinite (i.e., $t\rightarrow\infty$), it is
practical to investigate the probability of convergence in finite
time, rather than to make the auction convergent every time.
Fig.~\ref{fig:cdf} presents the cumulative
distribution functions (CDF) of the price converging iteration using
different step sizes. We observe that a small step size can raise
the probability of convergence. For example, to achieve the probability 1 of
convergence we need about $100$, $60$, and $2$ iterations for step
sizes $\epsilon_i=10^{-3}$, $\epsilon_i=10^{-4}$, and
$\epsilon_i=10^{-5}$, respectively. Therefore, if the step size is
small enough, our proposed mechanism can converge to the globally
optimal solution with probability 1. We can further see that, for
the same probability of convergence, node $1$ and node $2$ need more
iterations, the interpretation of which is addressed in the above
paragraph.

\begin{figure*}
 \hspace{-0.5cm}
  \begin{minipage}[t]{0.33\linewidth}
    \centering
    \includegraphics[scale=.43]{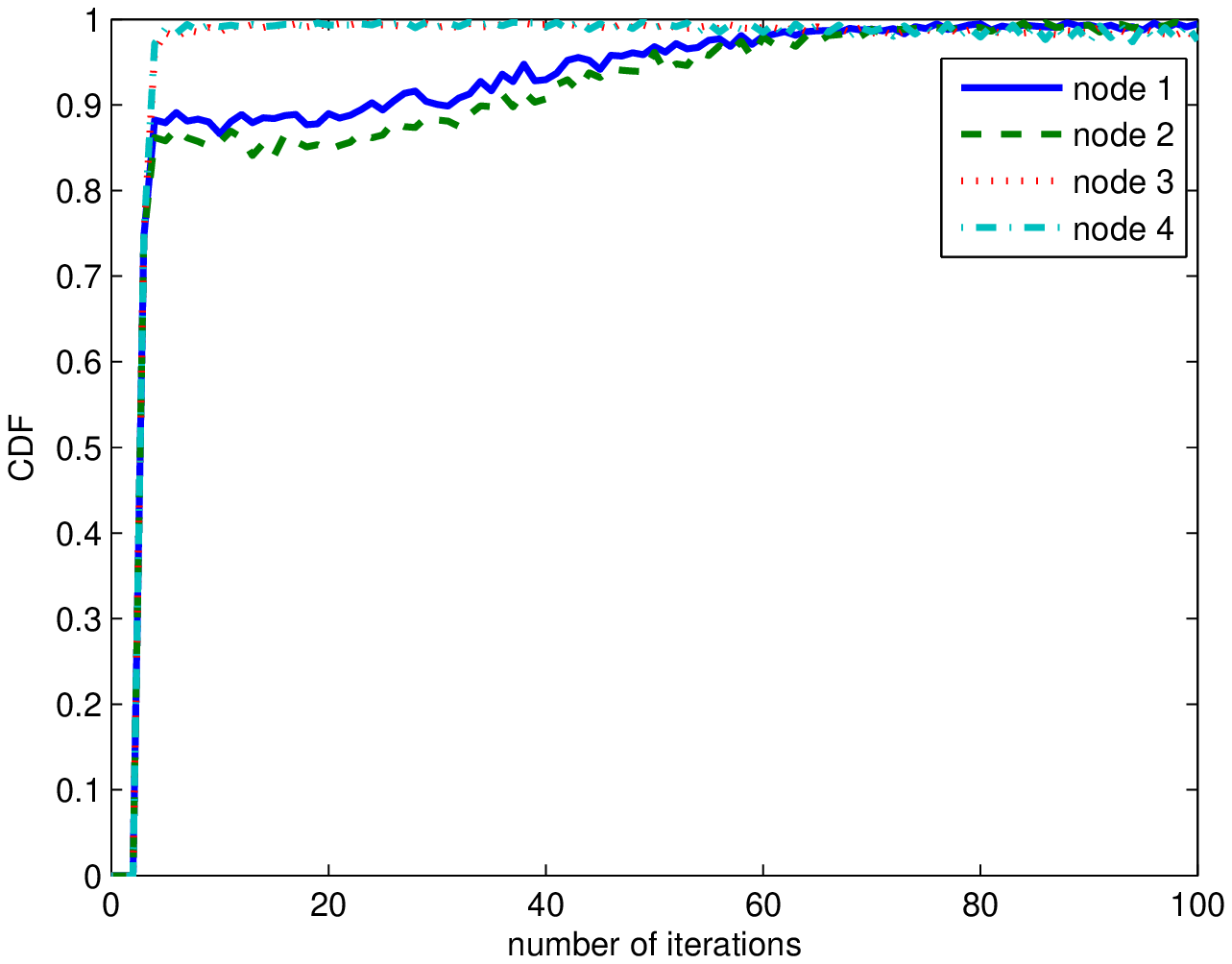}
    \end{minipage}
  \begin{minipage}[t]{0.01\linewidth}~
  \end{minipage}
  \begin{minipage}[t]{0.33\linewidth}
    \centering
    \includegraphics[scale=.43]{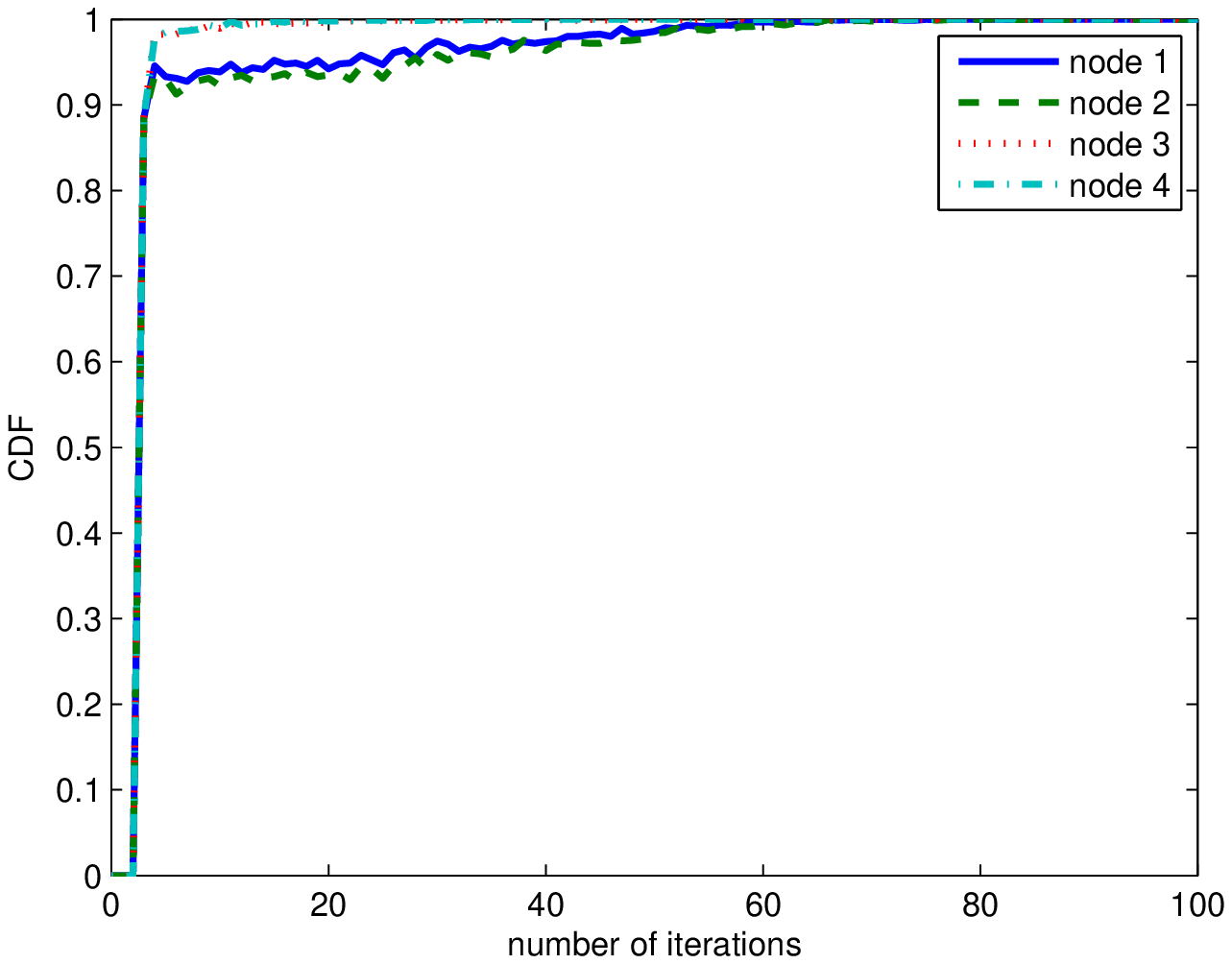}
  \end{minipage}
  \begin{minipage}[t]{0.01\linewidth}~
  \end{minipage}
  \begin{minipage}[t]{0.33\linewidth}
    \centering
\includegraphics[scale=.43]{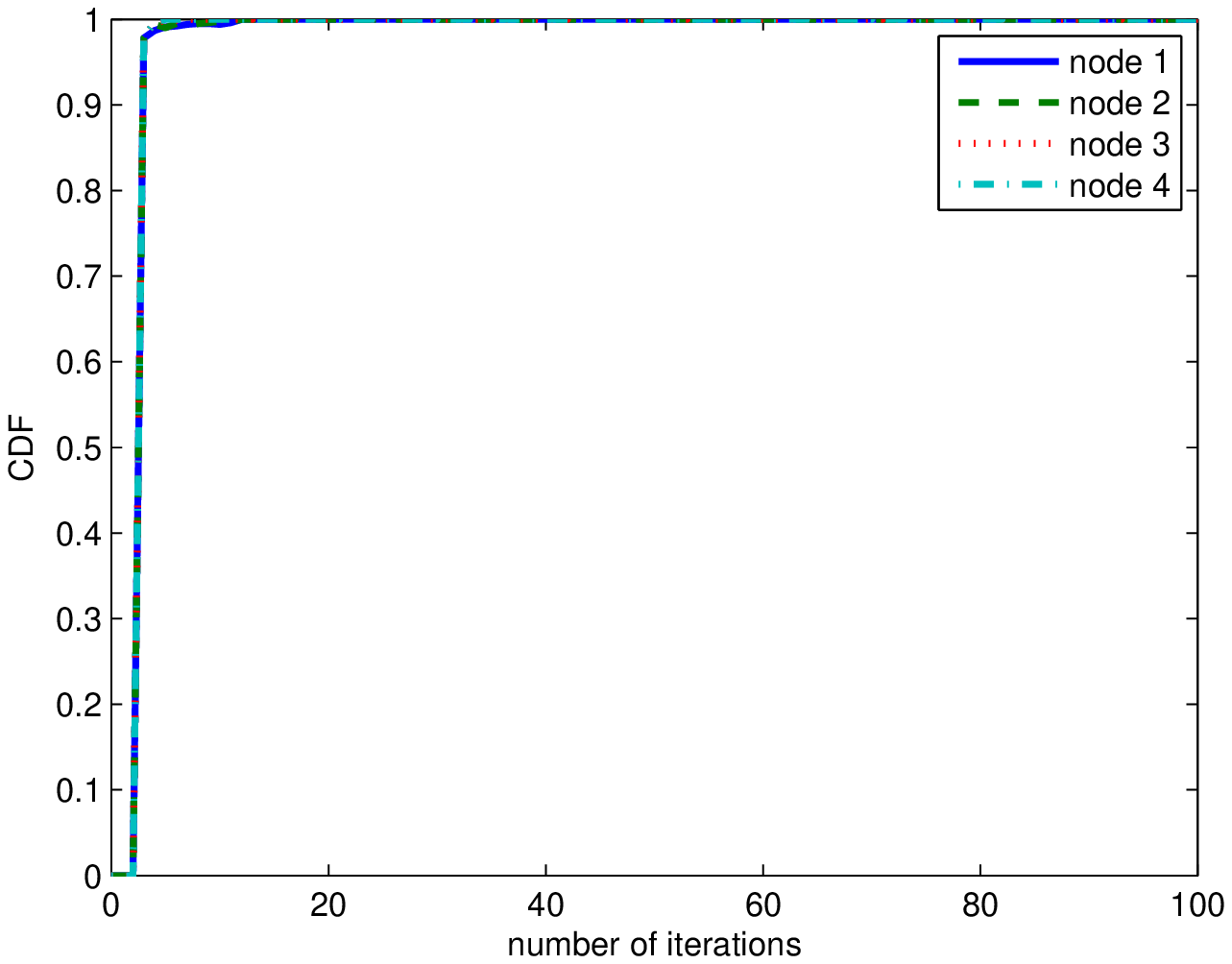}
\end{minipage}%
   \caption{CDF of price converging iterations with step sizes $\epsilon_i=10^{-3}$, $\epsilon_i=10^{-4}$ and $\epsilon_i=10^{-5}$, $\forall i$, respectively.}
   \label{fig:cdf}
\end{figure*}

\begin{figure*}
 \hspace{-0.5cm}
  \begin{minipage}[t]{0.33\linewidth}
    \centering
    \includegraphics[scale=.43]{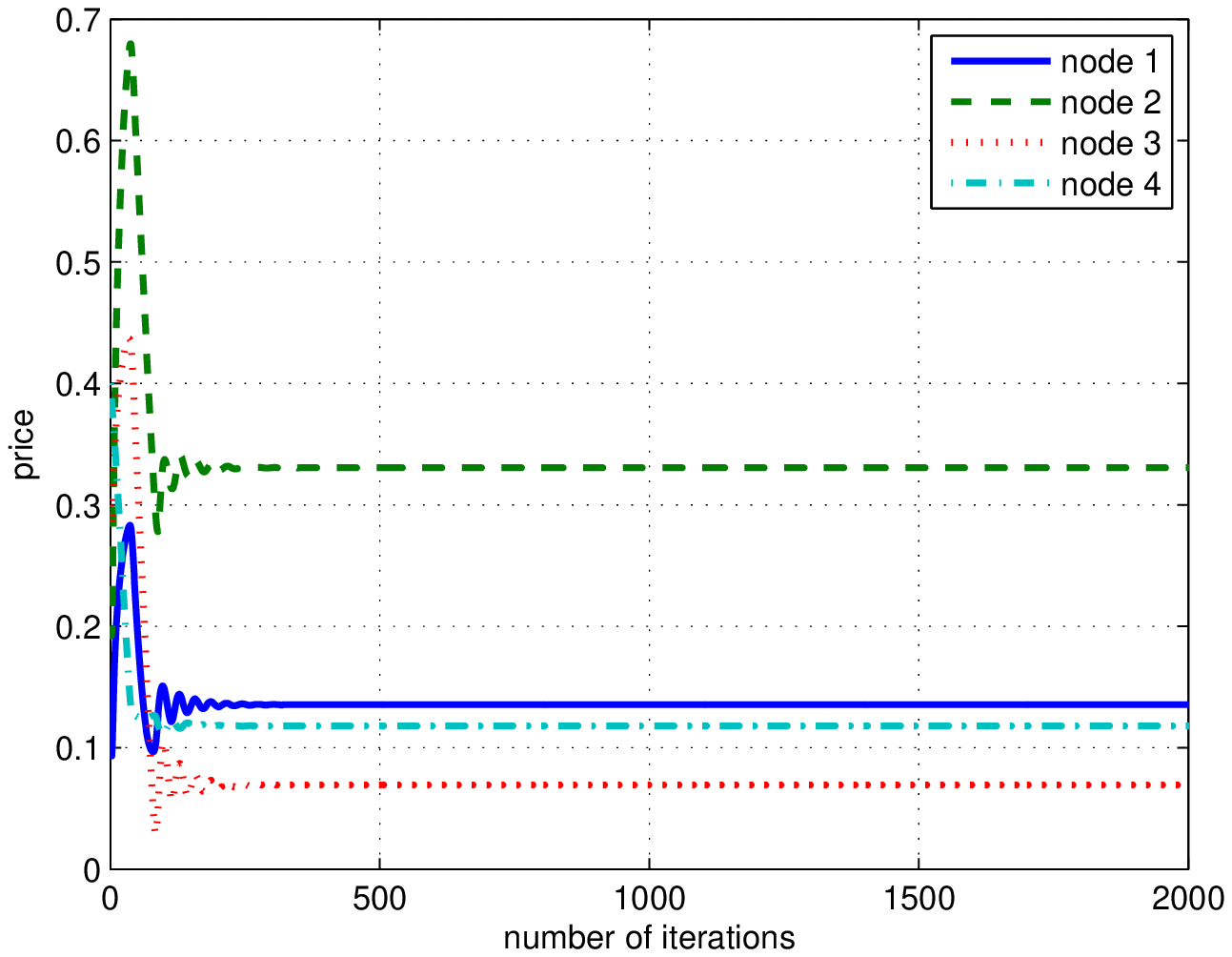}
    \end{minipage}
  \begin{minipage}[t]{0.01\linewidth}~
  \end{minipage}
  \begin{minipage}[t]{0.33\linewidth}
    \centering
    \includegraphics[scale=.43]{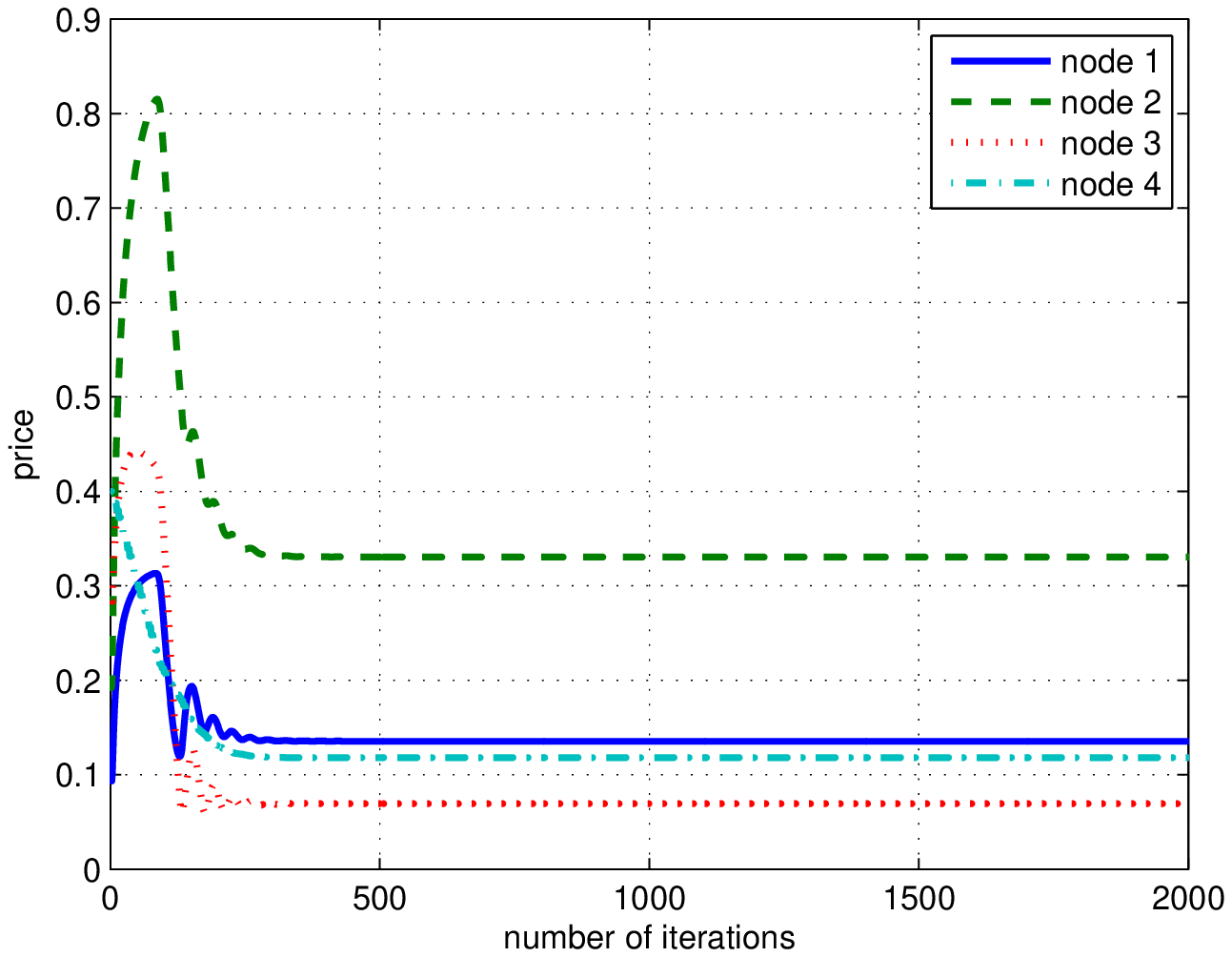}
  \end{minipage}
  \begin{minipage}[t]{0.01\linewidth}~
  \end{minipage}
  \begin{minipage}[t]{0.33\linewidth}
    \centering
\includegraphics[scale=.43]{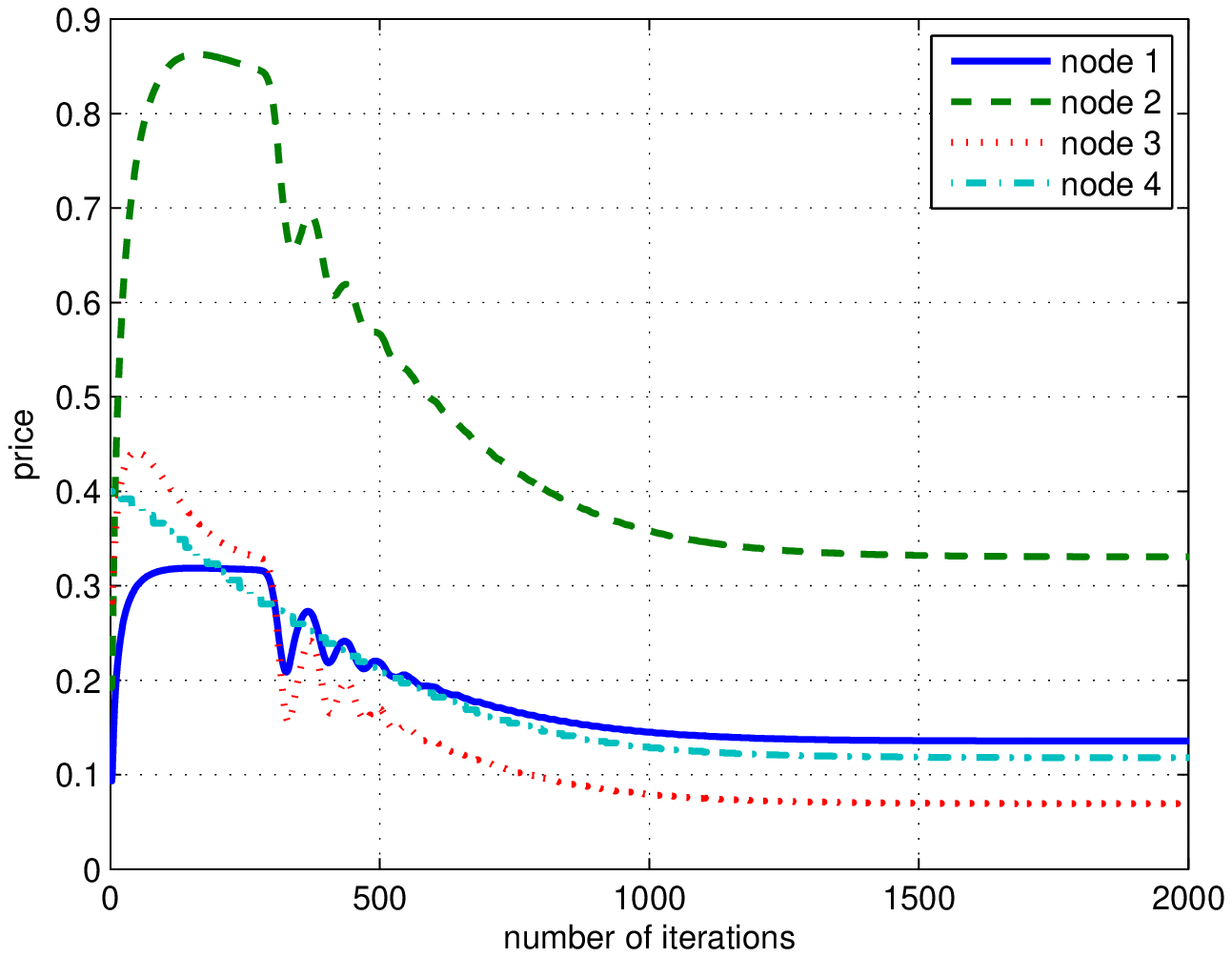}
\end{minipage}%
   \caption{Node 4's update frequencies are the same, $1/4$, and $1/20$ as others, respectively.}
   \label{fig:mal}
\end{figure*}

Finally, we consider the impacts of asynchronous updates among nodes in Fig.~\ref{fig:mal}, with the $10$ dB maximum power constraint and a constant step-size $\epsilon_i=10^{-3}$ for all nodes. Without loss of generality, it is assumed that node $4$ does not update frequently as others (e.g., malfunctional user). For a given channel realization, it needs about $200$ iterations for convergence if all nodes update their actions synchronously. When node $4$ updates slowly, the system convergence times becomes about $300$ and $1400$ iterations if node $4$'s update frequency is $1/4$ and $1/20$ of others, respectively. This shows that a slow updating node will slow down the overall convergence. Moreover, from the figure we observe that the proposed algorithm can converge to the global optimum, regardless of the slow updating node's update frequency (as long as it keeps updating). In other words, the asynchronous updates does not affect the network performance except the convergence time. The reason is that \textbf{P1} is essentially a convex optimization problem.

\subsection{Networks with More Users}

\begin{figure}[t]
\begin{centering}
\includegraphics[scale=.65]{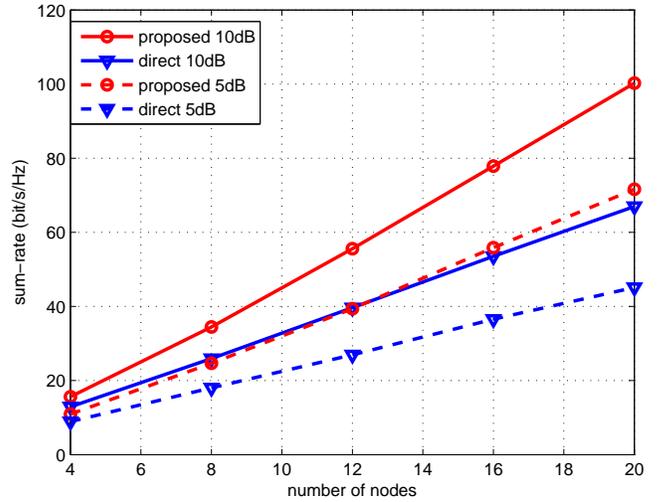}
\vspace{-0.1cm}
 \caption{Throughput comparison of two schemes.}\label{fig:throughput}
\end{centering}
\vspace{-0.3cm}
\end{figure}
\begin{figure}[t]
\begin{centering}
\includegraphics[scale=.55]{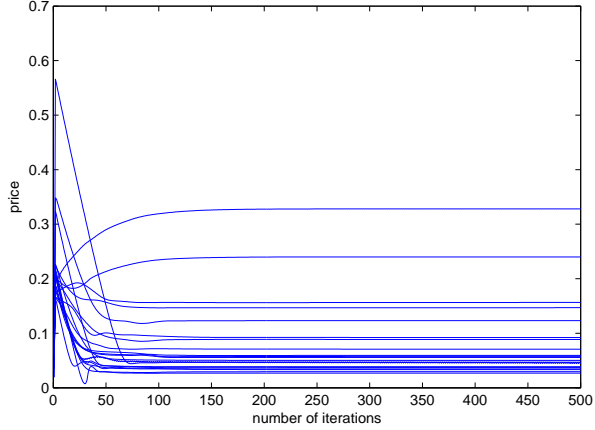}
\vspace{-0.1cm}
 \caption{Price convergence of $20$ nodes. Each line represents one node.}\label{fig:20user}
\end{centering}
\vspace{-0.3cm}
\end{figure}

In this subsection, we study the networks with more users. First, we compare the throughput performance of the proposed
power auction algorithm in comparison with the direct transmission
scheme in Fig.~\ref{fig:throughput}, with $\overline{p}_i=10$ dB and
$\overline{p}_i=5$ dB, $\forall i$, respectively. Here, a total of
$1000$ independent channel realizations are used. The locations of
nodes are random but uniformly distributed, varying with each
channel realization. The step size $\epsilon_i=10^{-3}$ is used,
$\forall i$. It is seen that the proposed algorithm outperforms the
direct transmission scheme by a significant margin, especially when
SNR is high or/and the number of nodes is large. This is consistent
with the previous study on cooperative communications from an
information theoretic perspective \cite{Laneman03,Laneman04}. Moreover, as
shown in Fig.~\ref{fig:20user}, we can observe that the complexity
of the proposed algorithm does not increase considerably even in the
case of large number of nodes.

\section{Conclusion and Future Work}

In this paper, we proposed a distributed framework for resource
allocation in cooperative networks.  We solved the problem by
mapping it into the multi-auctioneer multi-bidder power auction. By
following the proposed power auction, each user can capture the
optimal tradeoff of power allocation between the transmitter and
forwarder roles, and their behavior move towards the globally
optimal solution for weighted sum rates maximization. We further designed a
distributed realization mechanism to achieve the global optimum.
Our proposed framework can be generally applied to different classes
of networks, e.g., uplink cellular networks, ad hoc networks, and
peer-to-peer networks.

There are several directions to extend the results in this paper. This paper took the nonregenerative based AF relay strategy as an example when solving \textbf{P1}. We can also use other more advanced regenerative relay strategies (e.g., DF) or hybrid relay strategies for the similar problem, but the analysis will be more challenging. Moreover, this paper aimed to design distributed algorithms for solving the global optimization problem instead of game theoretical analysis. It will be interesting to further consider the incentive issues of the users (e.g., price-anticipating users) in a distributed network. Also, we assume that the message passing is timely and accurate. It will be very interesting to understand more about the performance of the algorithm under delayed and erroneous message passing (e.g.,\cite{Rad}) in our future work.

\appendices
\section{Proof of Proposition 1}
The proof is based on LaSalle's invariance principle \cite{Khalil}
and similar with that in \cite{Srikant}: Consider the differential
equation: $\dot{\boldsymbol{y}}(t)=f(\boldsymbol{y}(t))$. Let
$\boldsymbol{Y}: D\rightarrow\mathcal {R}$ be a radially
unbounded\footnote{$g(\boldsymbol{Z})$ is radially unbounded if
$\lim_{\|\boldsymbol{Z}\|\rightarrow\infty}g(\boldsymbol{Z})=\infty$.
}, continuously differentiable, positive
definite\footnote{$\boldsymbol{Y}$ is positive definite if
$\boldsymbol{Y}(\boldsymbol{Z}^*)=0$ for some $\boldsymbol{Z}^*$,
and $\boldsymbol{Y}(\boldsymbol{Z})>0$ for all $\boldsymbol{Z}\neq
\boldsymbol{Z}^*$.} function such that $\dot{\boldsymbol{Y}}\leq 0$
for all $\boldsymbol{Z}\in D$. Let $\mathcal {E}$ be the set of
points in $D$ where $\dot{\boldsymbol{Y}}=0$. Let $\mathcal {M}$ be
the largest invariant set\footnote{$\mathcal {M}$ is an invariant
set if $\boldsymbol{Z}(0)\in \mathcal {M}$ implies
$\boldsymbol{Z}(t)\in \mathcal {M}$ for all $t\geq 0.$} in $\mathcal
{E}$. Then every solution starting in $D$ approaches $\mathcal {M}$
as $t\rightarrow \infty$.

The Lyapunov function can be written as
\begin{equation}
V(\boldsymbol{\lambda}) = \frac{1}{2}\sum_{j\in \mathcal {K}}\left(
\lambda_j^{(t)} - \lambda_j^* \right)^2.
\end{equation}
It is obviously that this function is radially unbounded. We now
study time-derivative of this function. Differentiating this
function, we obtain:
\begin{eqnarray}
\dot{V}(\boldsymbol{\lambda}) &=& \sum_{j\in \mathcal {K}}\left(
\lambda_j^{(t)} - \lambda_j^* \right)\left(\sum_{i\in \mathcal
{K}}p_{j,i}^{(t)} - \overline{p}_j\right)_{\lambda_j}^+ \label{eqn:ieq1} \\
&\leq& \sum_{j\in \mathcal {K}}\left( \lambda_j^{(t)} - \lambda_j^*
\right)\left(\sum_{i\in \mathcal {K}}p_{j,i}^{(t)} -
\overline{p}_j\right) \label{eqn:ieq2}
\\ & = & \sum_{i\in \mathcal {K}}\sum_{j\in \mathcal {K}}\left( \lambda_j^{(t)} - \lambda_j^*
\right)\left(p_{j,i}^{(t)} - p_{j,i}^*\right) \nonumber \\
&& + \sum_{j\in \mathcal {K}}\left( \lambda_j^{(t)} - \lambda_j^*
\right)\left(\sum_{i\in \mathcal {K}}p_{j,i}^* -
\overline{p}_j\right) \label{eqn:ieq3} \\
&\leq& \sum_{i\in \mathcal {K}}\sum_{j\in \mathcal {K}}\left(
\lambda_j^{(t)} - \lambda_j^*
\right)\left(p_{j,i}^{(t)} - p_{j,i}^*\right) \label{eqn:ieq4}\\
& = & \sum_{i\in \mathcal {K}}\sum_{j\in \mathcal {K}}\left(
\frac{b_{j,i}^{(t)}}{p_{j,i}^{(t)}} - \frac{b_{j,i}^*}{p_{j,i}^*}
\right)\left(p_{j,i}^{(t)} - p_{j,i}^*\right)\label{eqn:ieq5}\\
&\leq& 0, \label{eqn:ieq6}
\end{eqnarray}
where (\ref{eqn:ieq2}) follows from (\ref{eqn:ieq1}) due to the
nature of  the projection $(a)_b^+$, i.e., if the projection is
active then (\ref{eqn:ieq1}) is zero while (\ref{eqn:ieq2}) is
positive, otherwise the equality holds. (\ref{eqn:ieq4}) follows due
to the fact that $\sum_{i\in \mathcal {K}}p_{j,i}^* =
\overline{p}_j$, or $\lambda_j^*=0$ if $\sum_{i\in \mathcal
{K}}p_{j,i}^* < \overline{p}_j$. Finally, if bids chose as $b_{j,i}
= p_{j,i} R'_i\left(p_{j,i}\right)$, (\ref{eqn:ieq6}) follows due to
the concavity and monotonicity of $R_i$.

Consequently, we have proved $\dot{V}(\boldsymbol{\lambda}) \leq 0$.
It also implies that
\begin{equation*}
\mathcal {E}:= \{\boldsymbol{\lambda}: \dot{V}(\boldsymbol{\lambda})
=0\}
\end{equation*}
is contained in the set
\begin{equation*}
\mathcal {S}:
\{\boldsymbol{\lambda}:(\ref{eqn:ieq1})=(\ref{eqn:ieq2})=(\ref{eqn:ieq4})=(\ref{eqn:ieq5})=0\}.
\end{equation*}

Let $\mathcal {M}$ be the largest invariant set of the primal-dual
algorithm contained in $\mathcal {E}$, then
$\boldsymbol{\lambda}^{(t)}$ converges to $\mathcal {M}$ as
$t\rightarrow \infty$. Since $\mathcal {M}\subset \mathcal
{E}\subset \mathcal {S}$, $\boldsymbol{\lambda}^{(t)}$ must satisfy
$\dot{V}(\boldsymbol{\lambda})=0$ as $t\rightarrow \infty$. Thus
$\lim_{t\rightarrow
\infty}\boldsymbol{\lambda}^{(t)}=\boldsymbol{\lambda}^*$.

Since $\boldsymbol{p}$ and $\boldsymbol{b}$ vary along with
$\boldsymbol {\lambda}$, according to LaSalle's invariance principle
\cite{Khalil}, it implies that $(\boldsymbol
{p}^{(t)},\boldsymbol{b}^{(t)}, \boldsymbol {\lambda}^{(t)})$
converges to the globally optimal point $(\boldsymbol
{p}^*,\boldsymbol {b}^*, \boldsymbol {\lambda}^*)$.

\bibliographystyle{IEEEtran}
\bibliography{IEEEabrv,Paper-TW-Dec-11-2284-auction}

\begin{thebibliography}{10}
\providecommand{\url}[1]{#1}
\csname url@samestyle\endcsname
\providecommand{\newblock}{\relax}
\providecommand{\bibinfo}[2]{#2}
\providecommand{\BIBentrySTDinterwordspacing}{\spaceskip=0pt\relax}
\providecommand{\BIBentryALTinterwordstretchfactor}{4}
\providecommand{\BIBentryALTinterwordspacing}{\spaceskip=\fontdimen2\font plus
\BIBentryALTinterwordstretchfactor\fontdimen3\font minus
  \fontdimen4\font\relax}
\providecommand{\BIBforeignlanguage}[2]{{%
\expandafter\ifx\csname l@#1\endcsname\relax
\typeout{** WARNING: IEEEtran.bst: No hyphenation pattern has been}%
\typeout{** loaded for the language `#1'. Using the pattern for}%
\typeout{** the default language instead.}%
\else
\language=\csname l@#1\endcsname
\fi
#2}}
\providecommand{\BIBdecl}{\relax}
\BIBdecl

\bibitem{YuanAuctionGC}
Y.~Liu, M.~Tao, and J.~Huang, ``Auction-based optimal power allocation in
  multiuser cooperative networks,'' in \emph{Proc. IEEE GLOBECOM}, 2011.

\bibitem{Laneman03}
J.~N. Laneman and G.~W. Wornell, ``Distributed space-time coded protocols for
  exploiting cooperative diversity in wireless networks,'' \emph{IEEE Trans.
  Inf. Theory}, vol.~49, no.~10, pp. 2415--2425, Oct. 2003.

\bibitem{Laneman04}
J.~N. Laneman, D.~N.~C. Tse, and G.W.Wornell, ``Cooperative diversity in
  wireless networks: efficient protocols and outage behavior,'' \emph{IEEE
  Trans. Inf. Theory}, vol.~50, no.~12, pp. 3062--3080, Dec. 2004.

\bibitem{Sendonaris1}
A.~Sendonaris, E.~Erkip, and B.~Aazhang, ``User cooperation diversity-part {I}:
  system description,'' \emph{IEEE Trans. Commun.}, vol.~51, no.~11, pp.
  1927--1938, Nov. 2003.

\bibitem{Sendonaris2}
------, ``User cooperation diversity-part {II}: implementation aspects and
  performance analysis,'' \emph{IEEE Trans. Commun.}, vol.~51, no.~11, pp.
  1939--1948, Nov. 2003.

\bibitem{Yu}
T.~Ng and W.~Yu, ``Joint optimization of relay strategies and resource
  allocations in cooperative cellular networks,'' \emph{IEEE J. Sel. Areas
  Commun.}, vol.~25, no.~2, pp. 328--339, Feb. 2007.

\bibitem{Himsoon}
T.~Himsoon, W.~Siriwongpairat, Z.~Han, and K.~J.~R. Liu, ``Lifetime
  maximization via cooperative nodes and relay deployment in wireless
  networks,'' \emph{IEEE J. Sel. Areas Commun.}, vol.~25, no.~2, pp. 306--317,
  Feb. 2007.

\bibitem{Tang}
J.~Tang and X.~Zhang, ``Cross-layer resource allocation over wireless relay
  networks for quality of service provisioning,'' \emph{IEEE J. Sel. Areas
  Commun.}, vol.~25, no.~4, pp. 645--656, May 2007.

\bibitem{Kim}
S.-J. Kim, X.~Wang, and M.~Madihian, ``Optimal resource allocation in multi-hop
  {OFDMA} wireless networks with cooperative relay,'' \emph{IEEE Trans.
  Wireless Commun.}, vol.~7, no.~5, pp. 1833--1838, May 2008.

\bibitem{Le}
L.~Le and E.~Hossain, ``Cross-layer optimization frameworks for multihop
  wireless networks using cooperative diversity,'' \emph{IEEE Trans. Wireless
  Commun.}, vol.~7, no.~7, pp. 2592--2602, Jul. 2008.

\bibitem{YuanTW}
Y.~Liu, M.~Tao, B.~Li, and H.~Shen, ``Optimization framework and graph-based
  approach for relay-assisted bidirectional {OFDMA} cellular networks,''
  \emph{IEEE Trans. Wireless Commun.}, vol.~9, no.~11, pp. 3490--3500, Nov.
  2010.

\bibitem{YuanTCOM}
Y.~Liu and M.~Tao, ``Optimal channel and relay assignment in {OFDM}-based
  multi-relay multi-pair two-way communication networks,'' \emph{IEEE Trans.
  Commun.}, vol.~60, no.~2, pp. 317--321, Feb. 2012.

\bibitem{Bletsas}
A.~Bletsas, A.~Khisti, D.~P. Reed, and A.~Lippman, ``A simple cooperative
  diversity method based on network path selection,'' \emph{IEEE J. Sel. Areas
  Commun.}, vol.~24, no.~3, pp. 659--672, Mar. 2006.

\bibitem{Savazzi}
S.~Savazzi and U.~Spagnolini, ``Energy aware power allocation strategies for
  multihop-cooperative transmission schemes,'' \emph{IEEE J. Sel. Areas
  Commun.}, vol.~25, no.~2, pp. 318--327, Feb. 2007.

\bibitem{Sergi}
S.~Sergi, F.~Pancaldi, and G.~M. Vitetta, ``A game theoretical approach to the
  management of transmission selection scheme in wireless ad-hoc networks,''
  \emph{IEEE Trans. Commun.}, vol.~58, no.~10, pp. 2799--2804, Oct. 2010.

\bibitem{Wang}
B.~Wang, Z.~Han, and K.~J.~R. Liu, ``Distributed relay selection and power
  control for multiuser cooperative communication networks using stackelberg
  game,'' \emph{IEEE Trans. Mobile Comput.}, vol.~8, no.~7, pp. 975--990, Jul.
  2009.

\bibitem{Huang}
J.~Huang, Z.~Han, M.~Chiang, and H.~V. Poor, ``Auction-based resource
  allocation for cooperative communications,'' \emph{IEEE J. Sel. Areas
  Commun.}, vol.~26, no.~7, pp. 1226--1237, Sep. 2008.

\bibitem{Ren}
S.~Ren and M.~van~der Schaar, ``Pricing and distributed power control in
  wireless relay networks,'' \emph{IEEE Trans. Signal Proc.}, vol.~59, no.~6,
  pp. 2913--2925, Jun. 2011.

\bibitem{Jing}
Y.~Jing and B.~Hassibi, ``Distributed space-time coding in wireless relay
  networks,'' \emph{IEEE Trans. Wireless Commun.}, vol.~5, no.~12, pp.
  3524--3536, Dec. 2006.

\bibitem{Yiu}
Y.~S and R.~Schober, ``Distributed space-time block coding,'' \emph{IEEE Trans.
  Commun.}, vol.~54, no.~7, pp. 1195 --1206, Jul. 2006.

\bibitem{Boyd}
S.~Boyd and L.~Vandenberghe, \emph{Convex Optimization}.\hskip 1em plus 0.5em
  minus 0.4em\relax Cambridge University Press, 2004.

\bibitem{Krishna}
V.~Krishna, \emph{Auction Theory}.\hskip 1em plus 0.5em minus 0.4em\relax
  Academic Press, 2002.

\bibitem{Kelly}
F.~P. Kelly, A.~K. Maulloo, and D.~K.~H. Tan, ``Rate control for communication
  networks: shadow prices, proportional fairness and stability,'' \emph{J.
  Oper. Res. Soc.}, vol.~49, no.~3, pp. 237--252, Mar. 1998.

\bibitem{Sun}
J.~Sun, E.~Modiano, and L.~Zheng, ``Wireless channel allocation using an
  auction algorithm,'' \emph{IEEE J. Sel. Areas Commun.}, vol.~24, no.~5, pp.
  1085--1096, May 2006.

\bibitem{Yang}
K.~Yang, N.~Prasad, and X.~Wang, ``An auction approach to resource allocation
  in uplink {OFDMA} systems,'' \emph{IEEE Trans. Signal Proce.}, vol.~57,
  no.~11, pp. 4482--4496, Nov. 2009.

\bibitem{Chen}
L.~Chen, S.~Iellamo, M.~Coupechoux, and P.~Godlewski, ``An auction framework
  for spectrum allocation with interference constraint in cognitive radio
  networks,'' in \emph{Proc. IEEE INFOCOM}, 2009.

\bibitem{Zhou}
X.~Zhou and H.~Zheng, ``{TRUST}: A general framework for truthful double
  spectrum auctions,'' in \emph{Proc. IEEE INFOCOM}, 2009.

\bibitem{Stanojev}
I.~Stanojev, O.~Simeone, U.~Spagnolini, Y.~B.-Ness, and R.~L. Pickholtz, ``An
  auction-based incentive mechanism for non-altruistic cooperative {ARQ} via
  spectrum-leasing,'' in \emph{Proc. IEEE GLOBECOM}, 2009.

\bibitem{Niyato1}
D.~Niyato, E.~Hossain, and Z.~Han, ``Dynamic spectrum access in {IEEE}
  802.22-based cognitive wireless networks: A game theoretic model for
  competitive spectrum bidding and pricing,'' \emph{IEEE Wireless Commun.},
  vol.~16, no.~2, pp. 16--23, Apr. 2009.

\bibitem{Niyato}
D.~Niyato and E.~Hossain, ``Competitive pricing for spectrum sharing in
  cognitive radio networks: Dynamic game, inefficiency of {Nash} equilibrium,
  and collusion,'' \emph{IEEE J. Sel. Areas Commun.}, vol.~26, no.~1, pp.
  192--202, Jan. 2008.

\bibitem{Huang2}
J.~Huang, R.~Berry, and M.~L. Honig, ``Auction-based spectrum sharing,''
  \emph{ACM Mobile Netw. Appl. J.}, vol.~11, no.~3, pp. 405--418, Jun. 2006.

\bibitem{Johari}
R.~Johari and J.~Tsitsiklis, ``Efficiency loss in a network resource allocation
  game,'' \emph{Math. Oper. Res.}, vol.~29, no.~3, pp. 407--435, Aug. 2004.

\bibitem{gradient}
\BIBentryALTinterwordspacing
S.~Boyd, L.~Xiao, and A.~Mutapcic, ``Subgradient methods,'' Stanford
  University, Oct. 2003. [Online]. Available: \url{[Lecture notes online
  available at http://www.stanford.edu/class/ee392o/subgrad-method.pdf]}
\BIBentrySTDinterwordspacing

\bibitem{Erceg}
V.~Erceg, L.~Greenstein, S.~Tjandra, S.~Parkoff, A.~Gupta, B.~Kulic, A.~Julius,
  and R.~Jastrzab, ``An empirically based path loss model for wireless channels
  in suburban environments,'' \emph{IEEE J. Sel. Areas Commun.}, vol.~17,
  no.~7, pp. 1205--1211, Jul. 1999.

\bibitem{Rad}
A.~H.~M. Rad, J.~Huang, M.~Chiang, and V.~W. Wong, ``Utility-optimal random
  access without message passing,'' \emph{IEEE Trans. Wireless Commun.},
  vol.~8, no.~3, pp. 1073--1079, 2009.

\bibitem{Khalil}
H.~Khalil, \emph{Nonlinear Systems}.\hskip 1em plus 0.5em minus 0.4em\relax
  Third Edition, Prentice Hall, 2002.

\bibitem{Srikant}
A.~Ozdaglar and R.~Srikant, \emph{Incentives and Pricing in Communication
  Networks}.\hskip 1em plus 0.5em minus 0.4em\relax Cambridge University Press,
  2007.

\end{thebibliography}

\begin{IEEEbiography}
{Yuan Liu}(S'11)
received the B.S. degree from Hunan University of Science and Technology, Xiangtan, China, in 2006, and the M.S. degree from Guangdong University of Technology, Guangzhou, China, in 2009, both in Communications Engineering and with the highest honors. He is currently pursuing his Ph.D. degree at the Department of Electrical Engineering in Shanghai Jiao Tong University. His current research interests include cooperative communications, network coding, resource allocation, physical layer security, MIMO and OFDM techniques.

He is the recipient of the Guangdong Province Excellent Master Theses Award in 2010. He has been honored as an Exemplary Reviewer of the \textsc{IEEE Communications Letters}. He is also awarded the
IEEE Student Travel Grant for IEEE ICC 2012. He is a student member of the IEEE.
\end{IEEEbiography}

\begin{IEEEbiography}
{Meixia Tao}(S'00-M'04-SM'10)
received the B.S. degree in electronic engineering from Fudan University, Shanghai, China, in 1999, and the Ph.D. degree in electrical and electronic engineering from Hong Kong University of Science and Technology in 2003. She is currently an Associate Professor with the Department of Electronic Engineering, Shanghai Jiao Tong University, China. From August 2003 to August 2004, she was a Member of Professional Staff at Hong Kong Applied Science and Technology Research Institute Co. Ltd. From August 2004 to December 2007, she was with the Department of Electrical and Computer Engineering, National University of Singapore, as an Assistant Professor. Her current research interests include cooperative transmission, physical layer network coding, resource allocation of OFDM networks, and MIMO techniques.

Dr. Tao is an Editor for the \textsc{IEEE Transactions on Communications} and the \textsc{IEEE Wireless Communications Letters}. She was on the Editorial Board of the \textsc{IEEE Transactions on Wireless Communications} from 2007 to 2011 and the \textsc{IEEE Communications Letters} from 2009 to 2012. She also served as Guest Editor for \textsc{IEEE Communications Magazine} with feature topic on LTE-Advanced and 4G Wireless Communications in 2012, and Guest Editor for \textsc{EURISAP J WCN} with special issue on Physical Layer Network Coding for Wireless Cooperative Networks in 2010. She was in the Technical Program Committee for various conferences, including IEEE INFOCOM, IEEE GLOBECOM, IEEE ICC, IEEE WCNC, and IEEE VTC.

Dr. Tao is the recipient of the IEEE ComSoC Asia-Pacific Outstanding Young Researcher Award in 2009.
\end{IEEEbiography}

\begin{IEEEbiography}
{Jianwei Huang}(S'01-M'06-SM'11)
is an Assistant Professor in the Department of Information Engineering at the Chinese University of Hong Kong. He received B.E. in Radio Engineering from Southeast University (Nanjing, Jiangsu, China) in 2000, M.S. and Ph.D. in Electrical and Computer Engineering from Northwestern University (Evanston, IL, USA) in 2003 and 2005, respectively. He worked as a Postdoc Research Associate in the Department of Electrical Engineering at Princeton University during 2005-2007. He was a visiting scholar at Ecole Polytechnique Federale De Lausanne (EPFL) and University of California-Berkeley.

Dr. Huang currently leads the Network Communications and Economics Lab (ncel.ie.cuhk.edu.hk), with the main research focus on nonlinear optimization and game theoretical analysis of networks, especially on network economics, cognitive radio networks, and smart grid. He has more than 100 publications in leading international journals, conferences, and books, which have been cited around 2000 times according to Google Scholar. He is the recipient of the IEEE Marconi Prize Paper Award in Wireless Communications in 2011, the International Conference on Wireless Internet Best Paper Award 2011, the IEEE GLOBECOM Best Paper Award in 2010, Asia-Pacific Conference on Communications Best Paper Award in 2009, the IEEE ComSoc Asia-Pacific Outstanding Young Researcher Award in 2009, and Walter P. Murphy Fellowship at Northwestern University in 2001.

Dr. Huang has served as one of the founding Editors of \textsc{IEEE Journal on Selected Areas in Communications} - Cognitive Radio Series, Editor of \textsc{IEEE Transactions on Wireless Communications} in the area of Wireless Networks and Systems, Guest Editor of IEEE \textsc{Journal on Selected Areas in Communications} special issue on "Economics of Communication Networks and Systems", Lead Guest Editor of \textsc{IEEE Journal of Selected Areas in Communications} special issue on "Game Theory in Communication Systems", Lead Guest Editor of \textsc{IEEE Communications Magazine} Feature Topic on "Communications Network Economics", and Guest Editor of several other journals including (Wiley) \textsc{Wireless Communications and Mobile Computing}, \textsc{Journal of Advances in Multimedia}, and \textsc{Journal of Communications}.

Dr. Huang has served as Chair of IEEE ComSocMMTC (Multimedia Communications Technical Committee) during 2012-2014, Steering Committee Member of \textsc{IEEE Transactions on Multimedia} and IEEE ICME during 2012 - 2014, Vice-Chair of IEEE MMTC during 2010 - 2012, Director of IEEE MMTC E-letter in 2010, and Chair of Meetings and Conference Committee of IEEE ComSoc Asia Pacific Board during 2012-2013. He also serves as the TPC Co-Chair of IEEE GLOBECOM Selected Aresa of Communications Symposium (Game Theory for Communications Track) 2013, the TPC Co-Chair of IEEE WiOpt (International Symposium on Modeling and Optimization in Mobile, Ad Hoc, and Wireless Networks) 2012, the Publicity Co-Chair of IEEE Communications Theory Workshop 2012, the TPC Co-Chair of IEEE ICCC Communication Theory and Security Symposium 2012, the Student Activities Co-Chair of IEEE WiOpt 2011, the TPC Co-Chair of IEEE GlOBECOM Wireless Communications Symposium 2010, the TPC Co-Chair of IWCMC (the International Wireless Communications and Mobile Computing) Mobile Computing Symposium 2010, and the TPC Co-Chair of GameNets (the International Conference on Game Theory for Networks) 2009. He is a TPC member of leading conferences such as INFOCOM (2009 - 2013), MobiHoc (2009, 2012), ICC, GLBOECOM, DySPAN, WiOpt, NetEcon, and WCNC. He is a senior member of the IEEE.
\end{IEEEbiography}

\end{document}